\documentclass[a4paper,11pt]{article}
\usepackage{amssymb,amsmath,latexsym,amsthm}
\usepackage{color,mathrsfs}
\usepackage{url}
\usepackage{enumerate}
\usepackage[square,sort,comma,numbers]{natbib}
\usepackage{dsfont}
\usepackage{diagbox}
 \usepackage{geometry}
\usepackage{graphicx}
\usepackage{colortbl}

\usepackage{endnotes}

\renewcommand{\textbf}[1]{\begingroup\bfseries\mathversion{bold}#1\endgroup}

\usepackage{fancyhdr}
\usepackage{pstricks,pst-plot,pst-node,pstricks-add}

\newtheorem{thm}{Theorem}[section]
\newtheorem{defi}[thm]{Definition}
\newtheorem{corollary}[thm]{Corollary}
\newtheorem{prop}[thm]{Proposition}

\newtheorem{lemma}[thm]{Lemma}

\theoremstyle{definition}
\newtheorem{remark}[thm]{Remark}

\newcommand{\R}{\mathbb R}

\newcommand{\Z}{\mathbb Z}
\newcommand{\N}{\mathbb N}

\numberwithin{equation}{section}

\def\XXint#1#2#3{{\setbox0=\hbox{$#1{#2#3}{\int}$}
    \vcenter{\hbox{$#2#3$}}\kern-.5\wd0}}

\allowdisplaybreaks
\date{date}
\begin{document}
\title{Minimal Soft Lattice Theta Functions}
\author{Laurent B\'{e}termin\\ \\
Faculty of Mathematics, University of Vienna,\\ Oskar-Morgenstern-Platz 1, 1090 Vienna, Austria\\ \texttt{laurent.betermin@univie.ac.at}. ORCID id: 0000-0003-4070-3344 }
\date\today
\maketitle

\begin{abstract}
We study the minimality properties of a new type of ``soft" theta functions. For a lattice $L\subset \R^d$, a $L$-periodic distribution of mass $\mu_L$ and an other mass $\nu_z$ centred at $z\in \R^d$, we define, for all scaling parameter $\alpha>0$, the translated lattice theta function $\theta_{\mu_L+\nu_z}(\alpha)$ as the Gaussian interaction energy between $\nu_z$ and $\mu_L$.  We show that any strict local or global minimality result that is true in the point case $\mu=\nu=\delta_0$ also holds for $L\mapsto \theta_{\mu_L+\nu_0}(\alpha)$ and $z\mapsto \theta_{\mu_L+\nu_z}(\alpha)$ when the measures are radially symmetric with respect to the points of $L\cup \{z\}$ and sufficiently rescaled around them (i.e. at a low scale). The minimality at all scales is also proved when the radially symmetric measures are generated by a completely monotone kernel. The method is based on a generalized Jacobi transformation formula, some standard integral representations for lattice energies and an approximation argument. Furthermore, for the honeycomb lattice $\mathsf{H}$, the center of any primitive honeycomb is shown to minimize $z\mapsto \theta_{\mu_{\mathsf{H}}+\nu_z}(\alpha)$ and many applications are stated for other particular physically relevant lattices including the triangular, square, cubic, orthorhombic, body-centred-cubic and face-centred-cubic lattices. 
\end{abstract}


\section{Introduction}

The mathematical justification of crystal's shape and formation is a very difficult problem which has been actively studied (see \cite{BlancLewin-2015} and references therein). Indeed, the analytic or numerical investigation of static many-particle Hamiltonian's ground states plays a central role in the design of materials (see e.g. \cite{Torquato06}), but the large number of critical points as well as the nonlinearities emerging from the corresponding systems make this mathematical investigation very challenging. Thus, a first natural step is to study systems that are already in a periodic order and where the interaction between points is given by a radially symmetric potential. These potentials arise in physics models of matter (see e.g. \cite{Kaplan,LikossoftCMatter}) in the case of the Born-Oppenheimer adiabatic approximation: the electrons effects are neglected, the energy is reduced to the nuclei interactions (see \cite[p. 33]{CondensMatter}), and two-body potentials are the simplest way to express the total potential energy of the system (see \cite[p. 945]{CondensMatter}).

\medskip

This problem of minimizing a potential energy per point of the form 
$$
E_f[L]:=\sum_{p\in L} f(|p|^2),
$$
where $L$ is a Bravais lattice (see also Definition \ref{def-Ef}), which is the most simple possible periodic configuration of points, has received a lot of attention, especially in the following cases: Lennard-Jones type potentials \cite{YBLB,BetTheta15,Beterminlocal3d,Beterloc}, Morse potential \cite{LBMorse}, two-dimensional Thomas-Fermi model for solid \cite{Betermin:2014fy}, Coulombian renormalized energy \cite{Sandier_Serfaty,SerfRoug15,Betermin:2014rr}, completely monotone interaction potentials \cite{CohnKumar,BetTheta15,CKMRV2Theta}, Bose-Einstein Condensates \cite{Mueller:2002aa,AftBN}, diblocks and 3-blocks copolymer interactions \cite{CheOshita,LuoChenWei}, vortices in quantum ferrofluids \cite{Vorticesferrofluids}, inverse power laws (Epstein zeta function) \cite{Rankin,Eno2,Cassels,Diananda,Coulangeon:kx,CoulLazzarini,CoulSchurm2018}, and also in more general settings \cite{Coulangeon:2010uq,OptinonCM}. An important mathematical object, which appears to be central in this theory (see e.g. \cite{CohnKumar,BetTheta15}), is the lattice theta function. Given a $d$-dimensional Bravais lattice $L$, a scaling parameter $\alpha>0$ and a point $z\in \R^d$, we define the lattice theta function and the translated lattice theta function by
\begin{equation}\label{defthetapoint}
\theta_L(\alpha):=\sum_{p\in L} e^{-\pi \alpha |p|^2}\quad \textnormal{and}\quad \theta_{L+z}(\alpha):=\sum_{p\in L} e^{-\pi \alpha |p+z|^2}.
\end{equation}
Physically, $\theta_L(\alpha)$ can be viewed as the Gaussian self-interaction of $L$ and $\theta_{L+z}(\alpha)$ as the Gaussian interaction between point $z$ and lattice $L$. They actually are the energies per point of the so-called Gaussian Core Model (GCM) restricted to lattices. This model was initially introduced by Stillinger \cite{Stillinger76} and motivated by the Flory-Krigbaum potential between the centers-of-mass of two polymer chains in an athermal solvent \cite{FloryKrigbaum}. The phase diagram of the three-dimensional GCM has been numerically investigated for example in \cite{TorquatoGCM08}. Furthermore, two different problems concerning \eqref{defthetapoint} appear to be quite natural once $\alpha>0$ is fixed: 
\begin{itemize}
\item the minimization of $L\mapsto \theta_L(\alpha)$ among $d$-dimensional Bravais lattices with the same density;
\item the minimization of $z\mapsto \theta_{L+z}(\alpha)$ among vectors $z\in \R^d$, where $L$ is fixed.
\end{itemize}
These minimization problems have been studied by many authors, see e.g. \cite{Mont,Baernstein-1997,SarStromb,CohnKumar,Coulangeon:2010uq,NierTheta,Regev:2015kq,BeterminPetrache,Faulhuber:2016aa,Beterminlocal3d,Cohn:2016aa,BeterminKnuepfer-preprint,FaulhuberExtremalDet,CoulSchurm2018,FalhuberOperatorNorms,FaulhuberBeaver,FaulhuberSteinerberger2,CKMRV2Theta}, and two of the most significant results are due to Montgomery \cite{Mont} -- who proved the minimality of the triangular lattice for $L\mapsto \theta_L(\alpha)$ among two-dimensional Bravais lattices of any fixed density -- and Cohn, Kumar, Miller, Radchenko and Viazovska \cite{CKMRV2Theta} -- who recently proved the minimality among all periodic configurations with the same density of $\mathsf{E}_8$ and the Leech lattice in dimensions $d\in \{8,24\}$. Furthermore, in \cite{BeterminKnuepfer-preprint}, the minimizer of the translated theta function $z\mapsto \theta_{L+z}(\alpha)$, for a fixed lattice $L$, has been also proved to be connected to the optimal electrostatic interaction (and more general long-range weighted interaction energies) between periodic distributions of charges located on $L$, solving a conjecture stated by Born in \cite{Born-1921} about the optimality of the rock-salt structure.

\medskip

Since polymer chains can be seen as soft interpenetrable spheres with an extent of the order of their radius of gyration (see \cite[Sect. 3]{LikossoftCMatter}), we propose a generalization of the periodic GCM to mass interactions, in the same spirit as it has been done in dimension $d=2$ in \cite{BetKnupfdiffuse} with Kn\"upfer. We want to study the minimality properties of lattices $L$ and points $z$ for a kind of ``soft GCM" (SGCM), where the objects are smeared out. Furthermore, these mass interaction energies can also be viewed as the expectation values of the lattice theta function and translated lattice theta function defined by \eqref{defthetapoint} where the position of the lattice points (resp. the position of $z$) follow a radially symmetric probability distribution $\mu_L$ (resp. $\nu_z$). We therefore define (see also Definition \ref{def-softlatticetheta}) the translated soft lattice theta function by
$$
\theta_{\mu_L+\nu_z}(\alpha):=\mathbb{E}_{\mu,\nu}[\theta_{L+z}(\alpha)]=\sum_{p\in L} \iint_{\R^d\times \R^d} e^{-\pi \alpha | x+p-z-y|^2}d\mu(x)d\nu(y).
$$ 
The main goal of this paper is to derive some minimality properties of lattices $L$ and points $z$ for $(L,z)\mapsto \theta_{\mu_L+\nu_z}(\alpha)$, where $\alpha>0$ is fixed, generalizing our previous work \cite{BetKnupfdiffuse} to the Gaussian interaction potential in all dimensions. In particular, the idea is to link the critical points, strict local minima or global minima of $(L,z)\mapsto \theta_{L+z}(\alpha)$ with those of $(L,z)\mapsto\theta_{\mu_L+\nu_z}(\alpha)$.

\medskip

Minimizing $L\mapsto \theta_{\mu_L+\nu_0}(\alpha)$ and $z\mapsto \theta_{\mu_L+\nu_z}(\alpha)$ can be interpreted in terms of "defects" in the periodic SGCM. Thus, we want to understand at which scales the type (i.e. the profile of $\mu,\nu$) or the size (i.e. the size of the ball containing almost all the mass of the measures) of the defects do not play any role in these minimization problems. Indeed, many defects appear in perfect crystals and they give generally to the material its properties (corrosion resistance, softness, thermal expansion, etc.). In Solid-State Physics, two kinds of point defects are important (see e.g. Kaxiras \cite[Chap. 9]{Kaxiras}):
\begin{enumerate}
\item the extrinsic defects, such as a substitutional impurity, corresponding to $z=0$ in our model. An atom in a perfect crystal is substituted by another one of the same kind ($\mu=\nu$) or of a different kind ($\mu\neq \nu$). This impurity is usually chemically similar to the crystal's atom, with a similar size. We are looking for the minimizer of $L\mapsto \theta_{\mu_L+\nu_0}(\alpha)$, i.e. the Gaussian interaction between the mass $\nu_0$ centred at the origin and all the masses $\mu_L$ centred at lattice sites, among a class of lattices with the same density. This also includes the energy per point of the perfect crystal itself when $\mu=\nu$ (see the top line of Figure \ref{Defects});
\item the intrinsic defects, such as an interstitial defect, corresponding to $z\neq 0$ in our model. An additional atom is located somewhere in the unit cell of the crystal (but not at a lattice site) and is generally smaller and chemically different than the crystal's atoms. We are looking for the minimizer of $z\mapsto \theta_{\mu_L+\nu_z}(\alpha)$, i.e. the Gaussian interaction between the mass $\nu_z$ centred at $z$ and all the masses $\mu_L$ centred at lattice sites (see  the bottom line of Figure \ref{Defects}).
\end{enumerate}

\begin{figure}[!h]
\centering
\includegraphics[width=8cm]{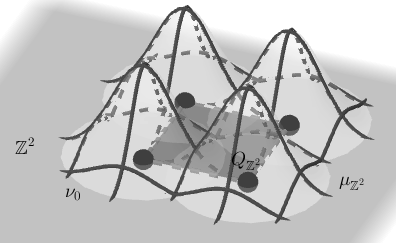} 
\includegraphics[width=8cm]{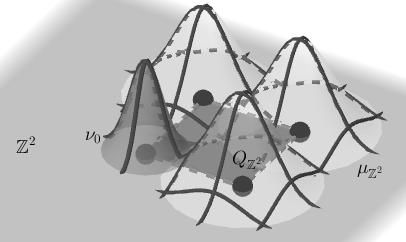}\\  \includegraphics[width=8cm]{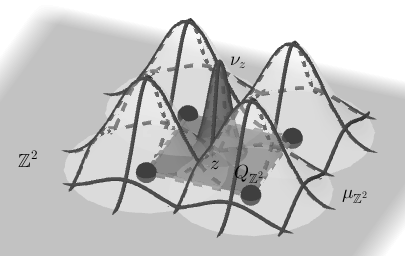} 
\caption{Example of different kinds of defects. The lattice is $L=\Z^2$ and the masses are all Gaussian, with different variances. We have chosen to represent only the primitive cell $Q_{\Z^2}$. On the top: the case without defect $\mu=\nu$ and $z=0$ (left), and the case of an extrinsic defect $\mu\neq \nu$ and $z=0$ (right). On the bottom: the case of a intrinsic defect $\mu\neq \nu$ and $z\neq 0$.}
\label{Defects}
\end{figure}

Using the methods developed in \cite{BetKnupfdiffuse} extended to the $d$-dimensional case and to the translated lattice theta function, we show that for any measures $\mu,\nu$ and any $\alpha>0$, we can rescale the measures around the points by sufficiently small factors $\varepsilon$ and $\delta$, getting two new measures $\mu^\varepsilon$ and $\nu^\delta$, such that any strict local or global minimality result which is true for $(L,z)\mapsto \theta_{L+z}(\alpha)$ is also true for $(L,z)\mapsto \theta_{\mu_L^\varepsilon+\nu^\delta_z}(\alpha)$. We will say that the minimality is true at a low scale.  In other words, there is a difference of scale between the lattice spacing and the radius of the balls where most of the masses are concentrated such that the measures (or the ``defects") do not play any role in terms of strict local or global minima. Furthermore, if $\mu$ and $\nu$ have densities of the form $x\mapsto \rho(|x|^2)$ where $\rho$ is the Laplace transform of a nonnegative Borel measure ($\rho$ is a completely monotone function), then the optimality occurs at all scales, i.e. for all such measures. As in \cite{BetKnupfdiffuse}, our results are purely qualitative: we do not give any value of $\varepsilon$ and $\delta$ such that the properties hold. However, it is interesting to see that, once $(\varepsilon,\delta)$ are below some threshold values $(\varepsilon_0,\delta_0)$, then the desired minimality occurs whatever the quotient $\varepsilon/\delta$ is. Furthermore, a numerical investigation have been made in \cite[Remark 14]{BetKnupfdiffuse} for the triangular lattice case with uniform measures on disks, showing that the strict local minimality of the triangular lattice proved for small values of the parameters $(\varepsilon,\delta)$ certainly does not hold when these parameters are too large.

\medskip

We notice that, as recalled in \cite{BetKnupfdiffuse}, this type of minimization problem involving smeared out particles appears in many physical and biological systems as condensed matter theory \cite{HeyesBranka}, quantum physics models \cite{BlancTFLattices}, diblock copolymer systems in the low volume fraction limit \cite{OhtaKawasaki}, magnetized disks interactions \cite{Disksinteractions} and swarming or flocking related models \cite{BalagueCarrilloEtal-2013,BernoffTopaz-2011,Carrilonowinteraction,TopazBertozzi-2004}.

\medskip

We have also devoted a complete section of the paper to the applications of our results. They mainly are corollaries of minimality results obtained in previous works in the point case $\mu=\nu=\delta_0$, i.e. for the lattice theta function and the translated lattice theta functions defined by \eqref{defthetapoint}, for some particular physically relevant lattices. We thought it was a good opportunity to review all those results in this paper in order to know what are the main open problems associated to these soft theta functions. Furthermore, we show that the minimum of $z\mapsto \theta_{\mathsf{H}+z}(\alpha)$, where $\mathsf{H}$ is a honeycomb lattice (see \eqref{def-honeycomb} and Figure \ref{Honeycomb}), is the centre of an honeycomb of $\mathsf{H}$, and the same is therefore true for the smeared out cases previously stated.

\medskip

In terms of generalization, it appears to be straightforward that all the results in this paper can be proved for more general energies of the form
$$
E_f[\mu_L+\nu_z]:=\sum_{p\in L} \iint_{\R^d\times \R^d} f(| x+p-z-y|^2)d\mu(x)d\nu(y),
$$
where $f$ is a  $L^1$ completely monotone summable function, as the one we have studied in \cite{BetKnupfdiffuse}. Since our original goal was to study a new kind of theta functions that could have other applications in Number Theory and Mathematical Physics, we did not extend our results to these types of energies. The reader can refer to \cite{BetKnupfdiffuse} for details. 

\medskip

\noindent \textbf{Plan of the paper.} After giving the definition of lattices, energies, measures and minimality at a low scale and at all scales in Section \ref{section-def}, we show some preliminary results in Section \ref{sect-prelim}, including the generalized Jacobi transformation formula. Our main results are stated and proved in Section \ref{sect-main} and many applications are given in Section \ref{sect-appli}, where the minimality of the primitive honeycomb's center in the honeycomb lattice case is also proved.
\section{Definitions}\label{section-def}

We start by defining the space of Bravais lattices with a fixed density as well as their unit cells and the notion of dual lattice. We call $(e_i)_{1\leq i\leq d}$ the orthonormal basis of $\R^d$, $|.|$ the euclidean norm on $\R^d$, $u\cdot v$ the associated scalar product of $u,v\in \R^d$ and $B_r$ the closed ball of radius $r>0$ and centred at the origin. Furthermore, we write $\mathcal{M}_d(\R)$ the space of $n\times n$ matrices with real coefficients.

\begin{defi}[Bravais lattice]
Let $d\geq 1$. We call $\mathcal{L}^\circ_d$ the space of $d$-dimensional Bravais lattices of the form $L=\bigoplus_{i=1}^d \Z u_i$ with basis $(u_1,...,u_d)\subset\R^d$ and covolume $1$, i.e. $\det(u_1,....,u_d)=1$. The unit cell (of volume 1) of such Bravais lattice $L$ is defined by 
\begin{equation*}
Q_L:=\left\{ x=\sum_{i=1}^d \lambda_i u_i \in \R^d, \lambda_i \in [0,1) \right\}.
\end{equation*}
Furthermore, the dual lattice of $L\in \mathcal{L}_d^\circ$ is defined by
$$
L^*:=\{x\in \R^d : \forall p\in L, x\cdot p\in \Z\}\in \mathcal{L}_d^\circ.
$$
\end{defi}

\noindent We also recall the definitions of the following important lattices belonging to $\mathcal{L}^\circ_d$ (see also Figure \ref{Lattices}):
\begin{align}
\textnormal{The triangular lattice}\quad &\Lambda_1:=\sqrt{\frac{2}{\sqrt{3}}}\left[\Z\left(1,0 \right)\oplus \Z\left( \frac{1}{2},\frac{\sqrt{3}}{2} \right)  \right];\label{Lambda1}\\
\textnormal{The (simple) cubic lattices}\quad &\Z^d,\quad d\geq 1;\label{cubic}\\
\textnormal{The orthorhombic lattices}\quad &\Z_a^d=\bigoplus_{i=1}^d \Z(a_i e_i), \quad \forall i,a_i>0,\quad \prod_{i=1}^d a_i=1;\label{ortho}\\
\textnormal{The Face-Centred-Cubic (FCC) lattice}\quad &\mathsf{D}_3:=2^{-\frac{1}{3}}\left[\Z(1,0,1)\oplus \Z(0,1,1)\oplus \Z(1,1,0)  \right];\label{FCC}\\
\textnormal{The Body-Centred-Cubic (BCC) lattice}\quad &\mathsf{D}_3^*:=2^{\frac{1}{3}}\left[\Z(1,0,0)\oplus \Z(0,1,0)\oplus \Z\left(\frac{1}{2},\frac{1}{2},\frac{1}{2}  \right)  \right]. \label{BCC}
\end{align}

\begin{figure}[!h]
\centering
\includegraphics[width=6cm]{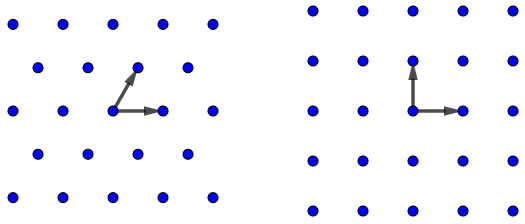} \\
\includegraphics[width=3cm]{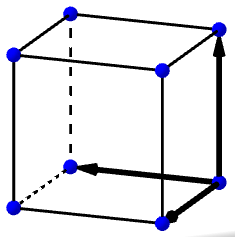} \quad \includegraphics[width=3cm]{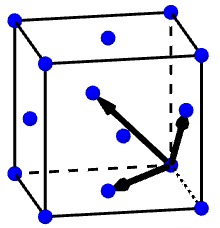}\quad \includegraphics[width=3cm]{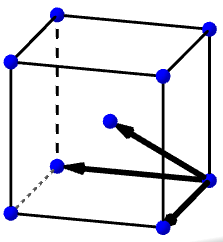}\\
\includegraphics[width=6cm]{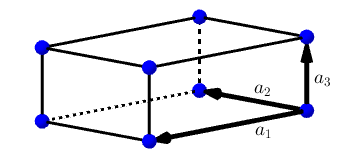}
 \caption{Representation of the triangular and square lattices $\Lambda_1,\Z^2$ (first line), the simple cubic, FCC and BCC lattices $\Z^3,\mathsf{D}_3,\mathsf{D}_3^*$ (second line), and the orthorhombic lattice $\Z_a^3$ (third line).}
\label{Lattices}
\end{figure}

These lattices are physically relevant because they correspond to the main crystal structures that exist in nature (see \cite[p. 8-9]{Kaxiras}) which are Bravais lattices. Notice for instance that among the 118 known elements, 21 (resp. 26) have a BCC (resp. FCC) structure. We will also consider the honeycomb lattice $\mathsf{H}$ defined by 
\begin{equation}\label{def-honeycomb}
\mathsf{H}:=\Lambda_1\cup(\Lambda_1+u),\quad u:=\sqrt{\frac{2}{\sqrt{3}}}\left(0,\frac{2}{\sqrt{3}}\right),
\end{equation}
and we call $\mathcal{H}$ one of its primitive hexagons (see Figure \ref{Honeycomb}). This is a typical example of periodic configuration, i.e. an union of translated Bravais lattices, that arises in Physics (e.g. as the structure of a graphene sheet).  Furthermore, we define $\mathsf{D}_4$, $\mathsf{E}_8$, $\mathsf{D}_d^+$ and the Leech lattice as in \cite{ConSloanPacking}. It turns out that these lattices have many interesting properties related to energy minimization or packing (see for instance \cite{SarStromb,Coulangeon:2010uq,CoulSchurm2018,Viazovska,CKMRV,CKMRV2Theta}).
\medskip

As explained in \cite[Sect. 1.4]{Terras_1988} (see also \cite[p. 14]{OptinonCM}), any lattice $L\in \mathcal{L}^\circ_d$ can be parametrized by a point $\bar{L}=(x_1,...,x_{n_d})\in\R^{n_d}$, where $n_d:=\frac{d(d+1)}{2}-1$, in a fundamental domain $\mathcal{D}_d$ where each lattice of $\mathcal{L}^\circ_d$ appears only once. Thus, as in \cite{BetKnupfdiffuse}, the metric on $\mathcal{D}_d$ is chosen as the euclidean metric on $\R^{n_d}$, $d(L,\Lambda)$ denotes the euclidean distance between two lattices $L,\Lambda\in \mathcal{D}_d$ and  $B_r(L)$ denotes the open ball of radius $r$ centred at $L$. We write $E\in \mathcal{C}^k(\mathcal{D}_d)$ if $E:\mathcal{D}_k\to \R$ is $k$-times differentiable with respect to the variables $(x_1,...,x_{n_d})$. We denote the gradient of $E$ by $\nabla E[L]:=(\partial_{x_1}E,...,\partial_{x_{n_d}}E)[L]$. The Hessian $D^2E \in \mathcal{M}_{n_d}(\R)$ is defined as the $n_d\times n_d$ real matrix of second derivatives with respect to $x_1,...,x_{n_d}$. $D^3 E$ is correspondingly the tensor of all third derivatives. The notions of strict local minimizer and critical point in $\mathcal{D}_d$ is defined as follows:
\begin{defi}
Let $d\geq 1$ and $E:\mathcal{D}_d\to \R$. We say that $L$ is a strict local minimizer of $E$ in $\mathcal{D}_d$ if there is $\eta>0$ such that $E[L]<E[\tilde{L}]$ for all $\tilde{L}\in B_\eta(L)$. Furthermore, $L$ is a critical point of $E$ in $\mathcal{D}_d$ if $\nabla E[L]=(0,...,0)$.
\end{defi}

Let us now focus on the interaction potentials we want to consider throughout this paper. Even though the main interaction potential will be the Gaussian one, it turns out that we will use many properties of more general lattice energies. We define two classes of functions that can be written as the Laplace transform of a measure.

\begin{defi}\label{def-classFdCMd}
Let $d\geq 1$. We say that $f\in \mathcal{F}_d$ if $|f(r)|=O(r^{-\frac{d}{2}-\sigma})$ as $r\to +\infty$ for some $\sigma>0$ and if $f$ can be represented as the Laplace transform of a Radon measure $\mu_f$, i.e.
$$
f(r)=\int_0^{+\infty} e^{-rt}d\mu_f(t).
$$
Furthermore we say that $f\in \mathcal{CM}_d$ if $f\in \mathcal{F}_d$ and $\mu_f$ is nonnegative.
\end{defi}

Furthermore, the energy per point of any $L\in \mathcal{D}_d$ interacting through a radial potential $f\in \mathcal{F}_d$ is defined by an absolutely convergent sum as follows.

\begin{defi}[Energy per point]\label{def-Ef}
Let $d\geq 1$, $L\in \mathcal{D}_d$ and $f\in \mathcal{F}_d$, then we define
\begin{equation}
E_f[L]:=\sum_{p\in L}f(|p|^2).
\end{equation}
\end{defi}

Because our goal is to study masses interactions, we need to specify what kind of measures we are working with. We note $\mathcal{P}(\R^d)$ the space of probability measures on $\R^d$ and $\mathcal{P}_r(\R^d)$ the space of probability measures on $\R^d$ that are rotationally symmetric with respect to the origin. Furthermore, we define the following subspace of $\mathcal{P}_r(\R^d)$:
\begin{equation*}
\mathcal{P}_r^{cm}(\R^d):=\left\{\mu\in \mathcal{P}_r(\R^d) : d\mu(x)=\rho(|x|^2)dx, \rho \in \mathcal{CM}_d   \right\}.
\end{equation*}

\begin{remark}[Completely monotone functions]
The notations $\mathcal{CM}_d$ and ``$cm$" mean that $\rho$ (or $f$ in Definition \ref{def-classFdCMd}) is a completely monotone function, i.e. $(-1)^k\rho^{(k)}(t)\geq 0$ for any $t>0$ and any $k\in \N$, which is indeed equivalent, by Hausdorff-Bernstein-Widder Theorem \cite{Bernstein-1929}, for $\rho$ to be the Laplace transform of a nonnegative Radon measure $\mu_\rho$.
\end{remark}

We now define the periodic measure $\mu_L$ that corresponds to the union of measures $\mu$ centred at all the points of a Bravais lattice $L$.

\begin{defi}[Periodized measure]
For any $L\in \mathcal{D}_d$ and any $\mu \in \mathcal{P}(\R^d)$, the periodized measure $\mu_L$ is defined by 
\begin{equation}\label{eq-periodmeasures}
\mu_L:=\sum_{p\in L} \mu_p,\quad \textnormal{where, for any $z\in \R^d$, $\mu_z:=\mu(\cdot -z)$}.
\end{equation}
\end{defi}
In Figure \ref{Measures}, we have represented two different kinds of measures $\mu_L$ for $\mu$ being in $\mathcal{P}_r(\R^2)$ and $\mathcal{P}_r^{cm}(\R^2)$. 
\medskip

\begin{figure}[!h]
\centering
\includegraphics[width=8cm]{MuLSelf.png} 
\includegraphics[width=8cm]{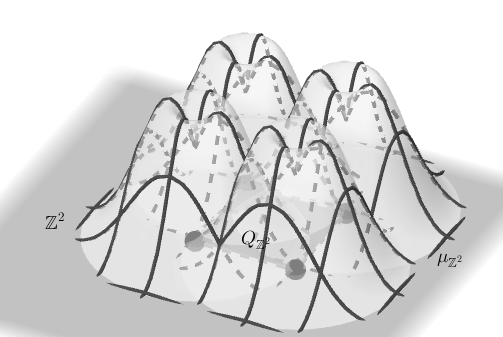}
\caption{Two kinds of periodic measures $\mu_L$, where $L=\Z^2$. The left-side one (resp. right-side one) is such that $\mu\in \mathcal{P}_r^{cm}(\R^2)$ (resp. $\mathcal{P}_r(\R^2)\backslash \mathcal{P}_r^{cm}(\R^2)$).}
\label{Measures}
\end{figure}

Furthermore, because we want to show the optimality of the previously defined lattices where the masses are sufficiently concentrated around the lattice sites, we define the rescaled measure as follows.

\begin{defi}[Rescaled measure]
Let $d\geq 1$. For any $\mu\in \mathcal{P}(\R^d)$ and any $\varepsilon>0$, the rescaled measure $\mu^\varepsilon$ is defined, for any measurable set $F\subset \R^d$, by
\begin{equation*}
\mu^\varepsilon(F):=\mu(\varepsilon F).
\end{equation*}
We write $\mu_L^\varepsilon$ and $\mu_z^\varepsilon$ the corresponding rescaled measures of $\mu_L$ and $\mu_z$ defined by \eqref{eq-periodmeasures}.
\end{defi}

We finally define the main energies that we are studying in this paper: the translated soft lattice theta function. It corresponds to the Gaussian interaction energy of the measure $\mu_L$ with itself or with an other measure $\nu_z$ located at $z\in \R^d$, and are therefore a generalization of the lattice and translated lattice theta functions $\theta_L(\alpha)$ and $\theta_{L+z}(\alpha)$ defined by \eqref{defthetapoint}.

\begin{defi}[Translated soft lattice theta functions]\label{def-softlatticetheta}
Let $d\geq 1$. For any $L\in \mathcal{D}_d$, any $\mu,\nu \in \mathcal{P}(\R^d)$, any $z\in \R^d$ and any $\alpha>0$, we define the translated soft lattice theta function of $\mu_L$ by the measure $\nu_z$ by
\begin{equation*}
\theta_{\mu_L+\nu_z}(\alpha):=\sum_{p\in L} \iint_{\R^d\times \R^d} e^{-\pi \alpha | x+p-z-y|^2}d\mu(x)d\nu(y).
\end{equation*}

\end{defi}

\begin{remark}[The honeycomb lattice case]
Since the honeycomb lattice $\mathsf{H}=\Lambda_1\cup (\Lambda_1+u)$ is a union of two translated triangular lattice, it is easy to compute its Gaussian energy per point. More precisely, we have
\begin{equation*}
\theta_{\mu_{\mathsf{H}}+\nu_z}(\alpha)=\frac{1}{2}\left(\theta_{\mu_{\Lambda_1}+\nu_z}(\alpha)+ \theta_{\mu_{\Lambda_1}+\nu_{z+u}}(\alpha)\right).
\end{equation*}
\end{remark}
\begin{remark}[Special cases]
We notice that:
\begin{itemize}
\item  For any $\mu\in \mathcal{P}(\R^d)$, any $L\in \mathcal{D}_d$ (or any periodic configuration) and any $\alpha>0$, it turns out that the self-interaction of $\mu_L$ is $\theta_{\mu_L + \mu_0}(\alpha)$ (i.e. when $\nu=\mu$ and $z=0$).
\item If $\mu=\nu=\delta_0$, then $\theta_{\mu_L+\nu_0}(\alpha)=\theta_L(\alpha)$ and $\theta_{\mu_L+\nu_z}(\alpha)=\theta_{L+z}(\alpha)$ as defined in \eqref{defthetapoint}.
\item The translated soft lattice theta function for the rescaled measures $\mu^\varepsilon$ and $\nu^\delta$ is given by
$$
\theta_{\mu^\varepsilon_{L}+\nu^\delta_z}(\alpha)= \sum_{p\in L} \iint_{\R^d\times \R^d} e^{-\pi \alpha | x+p-z-y|^2}d\mu^\varepsilon(x)d\nu^\delta(y).
$$
\end{itemize}
\end{remark}

%
The goal of this work is to study minimization problems for $L\mapsto \theta_{\mu_L+\nu_0}(\alpha)$ in $\mathcal{D}_d$ -- which includes the self-interaction of $\mu_L$ -- and $z\mapsto \theta_{\mu_L+\nu_z}(\alpha)$ in $Q_L$ for fixed $L$, at different scales. We therefore precise what we mean by being critical or minimal at a low scale and at all scales.

\begin{defi}[Criticality and minimality at all scales and at a low scale]
Let $d\geq 1$ and $\mu,\nu \in \mathcal{P}(\R^d)$, then:
\begin{enumerate}
\item We say that $L_0\in \mathcal{D}_d$ is a critical point or a (strict local) minimum of $L\mapsto \theta_{\mu_L+\nu_0}(\alpha)$ at all scales on $\mathcal{D}_d$  if, for any $\alpha>0$, $L_0$ is a critical point or a (strict local) minimum on $\mathcal{D}_d$ of $L\mapsto \theta_{\mu_L+\nu_0}(\alpha)$.
\item Let $\alpha_0>0$ be fixed. We say that $L_0\in \mathcal{D}_d$ is a critical point or a (strict local) minimum of $L\mapsto \theta_{\mu_L+\nu_0}(\alpha_0)$ at a low scale on $\mathcal{D}_d$ if there exist $\varepsilon_0>0$ and $\delta_0>0$ such that, for any $\varepsilon\in [0,\varepsilon_0)$ and any $\delta\in [0,\delta_0)$, $L_0$ is a critical point or a (strict local) minimum of $L\mapsto \theta_{\mu^\varepsilon_L+\nu^\delta_0}(\alpha_0)$ on $\mathcal{D}_d$.
\end{enumerate}
Furthermore, let $L$ be fixed, then we define the same notions of critical point and (strict local) minimality for $z\mapsto \theta_{\mu_L+\nu_z}(\alpha)$ of $z_0$ in $Q_L$ at all scales or at a low scale.
\end{defi}



\section{Preliminaries}\label{sect-prelim}

It is well-known that the lattice theta function satisfies the following identity, called Jacobi Transformation Formula (see e.g. \cite[Thm. A]{SaffLongRange} or \cite{Bochnertheta} for a general formula involving harmonic polynomials): for any $L\in \mathcal{D}_d$, any $z\in \R^d$ and any $\alpha>0$,
\begin{equation}\label{eq-jacobipoint}
\theta_{L+z}(\alpha):=\sum_{p\in L} e^{-\pi \alpha |p+z|^2}=\frac{1}{\alpha^{\frac{d}{2}}}\sum_{q\in L^*} e^{-\frac{\pi |q|^2}{\alpha}}e^{2i \pi q\cdot z}.
\end{equation}

\noindent We generalize this formula to mass interaction in the following result.

\begin{prop}[Generalized Jacobi Transformation Formula]\label{prop-Jacobigeneral}
For any $d\geq 1$, any $L\in \mathcal{D}_d$, any $\mu,\nu\in \mathcal{P}_r(\R^d)$, any $z\in \R^d$ and any $\alpha>0$, we have
\begin{equation}\label{Jacobigeneral}
\theta_{\mu_L+\nu_z}(\alpha) =\frac{\Gamma\hspace{-1mm} \left( \frac{d}{2} \right)^2}{\alpha^{\frac{d}{2}}} \sum_{q\in L^*} \frac{e^{-\frac{\pi |q|^2}{\alpha}}g_\mu(|q|)g_\nu(|q|)}{|q|^{d-2}}e^{2i \pi q\cdot z},
\end{equation}
where, for any $m\in \mathcal{P}_r(\R^d)$, 
\begin{equation*}
g_m(|q|):= \int_0^\infty J_{\frac{d}{2}-1}(4\pi s |q|)s^{1-\frac{d}{2}}d\psi_m(s),
\end{equation*}
and $\psi_m$ is the Lebesgue-Stieltjes measure of $t\mapsto m(B_t)$, i.e. $\psi_m([r_1,r_2))=m(B_{r_2})-m(B_{r_1})$ and $J_\beta$ is the Bessel function of the first kind.
\end{prop}
\begin{proof}
By the classical Jacobi Transformation Formula \eqref{eq-jacobipoint} and Fubini's Theorem, we get
\begin{align*}
\theta_{\mu_L+\nu_z}(\alpha) &=\iint_{\R^2\times \R^2} \sum_{p\in L} e^{-\pi\alpha |p+x-z-y|^2}d\mu(x)d\nu(y)\\
&=\frac{1}{\alpha^{\frac{d}{2}}}\iint_{\R^2\times \R^2} \sum_{q\in L^*} e^{2 i \pi q \cdot (x-y-z)}e^{-\frac{\pi|q|^2}{\alpha}}d\mu(x)d\nu(y)\\
&= \frac{1}{\alpha^{\frac{d}{2}}} \sum_{q\in L^*}e^{-\frac{\pi|q|^2}{\alpha}}e^{-2i \pi q\cdot z}\left(\int_{\R^2} e^{2i\pi q\cdot x}d\mu(x)\right)\left(\int_{\R^2} e^{-2i\pi q\cdot y}d\nu(y)\right)\\
&=\frac{1}{\alpha^{\frac{d}{2}}} \sum_{q\in L^*}e^{-\frac{\pi|q|^2}{\alpha}}e^{-2i \pi q\cdot z} \hat{\mu}(2q)\hat{\nu}(2q),
\end{align*}
where $\hat{m}$ is the notation for the Fourier transform of a measure $m\in \mathcal{P}_r(\R^d)$. We now recall that $\hat{m}$ is given by the Hankel-Stieltjes transform (see e.g. \cite[Section 2]{ConnetSchwartz}), i.e. for any $x\in \R^d$, 
\begin{equation*}
\hat{m}(x)=\int_{\R^2} e^{-i \pi x\cdot y}dm(y)=\frac{2^{\frac{d}{2}-1}\Gamma\left(\frac{d}{2} \right)}{|x|^{\frac{d}{2}-1}}\int_0^{\infty} J_{\frac{d}{2}-1}(2\pi s |x|)s^{1-\frac{d}{2}}d\psi_m(s),
\end{equation*}
where $\psi_m$ is the Lebesgue-Stieltjes measure of $t\mapsto m(B_t)$, and the proof is completed.
\end{proof}
\begin{remark}[Value at the origin and connection with our previous work]
The value for $q=0$ is well-defined. Indeed, using the Taylor expansion of $J_{\frac{d}{2}-1}$ (see e.g. \cite[Eq. (9.1.30)]{AbraStegun}) and the fact that $\psi_m\in \mathcal{P}(\R_+)$, we obtain
$$
\lim_{q\to 0} \Gamma\hspace{-1mm}\left( \frac{d}{2}\right)^2\frac{e^{-\frac{\pi |q|^2}{\alpha}}g_\mu(|q|)g_\nu(|q|)}{|q|^{d-2}}e^{2i \pi q\cdot z}=1.
$$
We furthermore notice that for $d=2$, $z=0$ and $\mu=\nu$, we recover the formula proved in \cite[Prop. 7]{BetKnupfdiffuse}, up to a factor $2$ in the argument of $J_{\frac{d}{2}-1}$ that we have corrected and which does not influence the final result.
\end{remark}

\noindent We now give some basic facts that will be used in the proofs of our main results.

\begin{lemma}\label{lem-L1}
Let $d\geq 1$. If $f,g\in \mathcal{F}_d$, then $fg\in \mathcal{F}_d$.
\end{lemma}
\begin{proof}
If $f$ and $g$ have both this representation as the Laplace transform of a measure, the formula $\mathcal{L}^{-1}[f]\ast \mathcal{L}^{-1}[g]=\mathcal{L}^{-1}[fg]$ (see e.g. \cite[Thm. 5.3.11]{Poularikas}) shows that $fg$ has also this representation. Moreover, we naturally have $|f(r)g(r)|=O(r^{-d/2-\sigma})$ for some $\sigma>0$ as $r\to +\infty$, because it is the case for both functions.
\end{proof}

\noindent The next lemma is a simple consequence of the Jacobi Transformation Formula \eqref{eq-jacobipoint}.

\begin{lemma}\label{lem-L2}
Let $d\geq 1$ and $\alpha>0$, then $L_0$ is a critical point or a (strict local) minimizer of $L\mapsto \theta_L(\alpha)$ on $\mathcal{D}_d$  if and only if $L_0^*$ is a critical point or a (strict local) minimizer  of $L\mapsto \theta_L(1/\alpha)$ on $\mathcal{D}_d$ . Furthermore, if $L_0$ is the unique global minimizer of $L\mapsto \theta_L(\alpha)$ on $\mathcal{D}_d$  for all $\alpha>0$, then $L_0=L_0^*$.
\end{lemma}

\noindent The following results come from the simple fact that, for any $f\in \mathcal{F}_d$, any $L\in \mathcal{D}_d$ and any $z\in \R^d$,
$$
E_f[L+z]:=\sum_{p\in L} f(|p+z|^2)=\int_0^{+\infty}\theta_{L+z}\left(\frac{t}{\pi}\right)d\mu_f(t),
$$
as explained e.g. in \cite[Sect. 3.1]{BetTheta15}.

\begin{lemma}\label{lem-L3}
Let $d\geq 1$. If $L_0$ is a critical point of $L\mapsto\theta_L(\alpha)$ in $\mathcal{D}_d$ for all $\alpha>0$, then, for any $f\in \mathcal{F}_d$, $L_0$ and $L_0^*$ are critical points of $L\mapsto E_f[L]$ in $\mathcal{D}_d$.
\end{lemma}

\begin{lemma}\label{lem-L4}
Let $d\geq 1$. If $L_0$ is a (strict local) minimizer in $\mathcal{D}_d$ of $L\mapsto \theta_L(\alpha)$ for all $\alpha>0$, then, for any $f\in \mathcal{CM}_d$, $L_0$ is a (strict local) minimizer of $L\mapsto E_f[L]$ in $\mathcal{D}_d$.\\
Let $L\in \mathcal{D}_d$. If $z_0$ is  a (strict local) minimizer of $z\mapsto \theta_{L+z}(\alpha)$ in $Q_L$ for all $\alpha>0$, then, for any $f\in \mathcal{CM}_d$, $z_0$  is a (strict local) minimizer of $z\mapsto E_f[L+z]$ in $Q_L$.
\end{lemma}

\section{Main results}\label{sect-main}

In this part, we generalize the results of \cite{BetKnupfdiffuse} about the (strict local) minimality of a lattice for $L\mapsto \theta_{\mu_L+\nu_0}(\alpha)$, for a given $\alpha >0$, to any dimension. Furthermore, using the same strategy based on an approximation of the quantity we sum by a completely monotone potential, we show the same kind of results for $z\mapsto \theta_{\mu_L+\nu_z}(\alpha)$ where $L$ is fixed.

\medskip

It is important to notice that, by the generalized Jacobi Transformation Formula \eqref{Jacobigeneral}, if $z=0$, then $\theta_{\mu_L+\nu_0}(\alpha)$  is the sum of a radial potential over $L^*$, i.e. an energy of type $E_f[L^*]$ for some $f$, as defined in Definition \ref{def-Ef}. Therefore, as shown for instance in \cite{Coulangeon:2010uq,Beterminlocal3d,Beterloc}, any lattice $L$ such that $L^*$ has enough symmetries is a critical point of the soft lattice theta function, e.g. $\Lambda_1$, $\mathsf{D}_3$, $\mathsf{D}_3^*$ and $\Z^d$, $d\geq 2$. It turns out that, as proved in \cite[Thm. 3.2]{LBMorse} these lattices are the only one in dimensions $d\in \{2,3\}$ being ``volume-stationary" for an energy of type $E_f$, i.e. they can be critical point of $E_f$ in the space of Bravais lattices of fixed density in an open interval of densities. The following result gives some sufficient conditions for a lattice or a point, which are already critical for the theta functions \eqref{defthetapoint}, to be critical for the soft lattice theta functions.

\begin{prop}[Criticality]\label{prop-criticality}
Let $d\geq 1$ and $\mu,\nu\in \mathcal{P}_r(\R^d)$ then the following hold:
\begin{enumerate}
\item  If $L_0$ is a critical point of $L\mapsto \theta_L(\alpha)$ in $\mathcal{D}_d$ for all $\alpha>0$, then $L_0$ is a critical point of $L\mapsto \theta_{\mu_L+\nu_0}(\alpha)$ in $\mathcal{D}_d$ for all $\alpha>0$.
\item Let $L\in \mathcal{D}_d$.  If $z_0$ is a critical point of $z\mapsto \theta_{L+z}(\alpha)$ in $Q_L$ for all $\alpha>0$, then $z_0$ is a critical point of $z\mapsto \theta_{\mu_L+\nu_z}(\alpha)$ in $Q_L$ for all $\alpha>0$.
\end{enumerate}
\end{prop}
\begin{proof}
Let us prove the first part of the proposition. By Lemma \ref{lem-L3}, if $L_0$ is a critical point of $L\mapsto \theta_L(\alpha)$ in $\mathcal{D}_d$ for all $\alpha>0$, then $L_0$ and $L_0^*$ are critical points of $E_f$ in $\mathcal{D}_d$ where $f\in \mathcal{F}_d$. We claim that there exists a Radon measure $\mu_h$ such that the function $h$ defined by $h(r):=\frac{e^{-\frac{\pi r}{\alpha}}g_\mu(\sqrt{r})g_\nu(\sqrt{r})}{r^{\frac{d}{2}-1}}$ belongs to the set $\mathcal{F}_d$, i.e. $h=\mathcal{L}[\mu_h]$. It is clear by Lemma \ref{lem-L1} because $h$ is a product of four functions belonging to $\mathcal{F}_d$, which completes the proof because $L_0$ is hence a critical point of $E_h$ in $\mathcal{D}_d$.

\medskip

Let us prove the second part of the proposition. Since $z_0$ is a critical point of $z\mapsto \theta_{L+z}(\beta)$ for all $\beta>0$, we obtain, using Jacobi transformation formula \eqref{eq-jacobipoint} and computing the derivative with respect to $z_i$,
\begin{equation}\label{eq-criticpointpoint}
\forall i\in \{1,...,d\}, \quad \sum_{q\in L^*} e^{-\frac{\pi |q|^2}{\beta}}q_i \sin(2\pi q\cdot z_0)=0.
\end{equation}
We recall that, by Proposition \ref{prop-Jacobigeneral}, the soft lattice theta function can be written as
\begin{equation*}
\theta_{\mu_L+\nu_z}(\alpha) =\frac{\Gamma\hspace{-1mm}\left( \frac{d}{2} \right)^2}{4\alpha^{\frac{d}{2}}} \sum_{q\in L^*} h(|q|^2)\cos(2\pi q\cdot z),\quad h(r):=\frac{e^{-\frac{\pi r}{\alpha}}g_\mu(\sqrt{r})g_\nu(\sqrt{r})}{r^{\frac{d}{2}-1}}.
\end{equation*}
We now use the fact that there exists a Borel measure $\mu_h$ such that $h=\mathcal{L}[\mu_h]\in \mathcal{F}_d$ as explained above. Therefore, integrating \eqref{eq-criticpointpoint} against $\mu_h$, where $t=\pi/\beta$ is the variable of integration, gives
\begin{equation*}
\forall i\in \{1,...,d\}, \quad \sum_{q\in L^*} h(|q|^2) q_i \sin(2\pi q\cdot z_0)=0,
\end{equation*}
for any $\alpha>0$. It follows that $\partial_{z_i}\theta_{\mu_L+\nu_{z_0}}(\alpha)=0$ for all $ i\in \{1,...,d\}$ and all $\alpha>0$, i.e. $z_0$ is a critical point of $z\mapsto \theta_{\mu_L+\nu_z}(\alpha)$ in $Q_L$ for all $\alpha>0$. 
\end{proof}

\begin{remark}[Critical points of $L\mapsto \theta_L(\alpha)$ for all $\alpha>0$] It turns out that the only lattices $L_0\in \mathcal{D}_d$ that can be critical point of $L\mapsto \theta_L(\alpha)$ for all $\alpha>0$ in dimensions $d\in \{2,3\}$ are $\Z^2$, $\Lambda_1$, $\Z^3$, $\mathcal{D}_3$ and $\mathcal{D}_3^*$, as proved in \cite[Section 3]{LBMorse}. 

\end{remark}

The next result generalizes \cite[Prop. 11]{BetKnupfdiffuse} to $L\mapsto \theta_{\mu_L+\nu_0}(\alpha)$ and $z\mapsto \theta_{\mu_L+\nu_z}(\alpha)$ by giving a sufficient condition for the strict local minimality of a lattice or a point, which are already minimal for the theta functions \eqref{defthetapoint}, in the case of masses interactions.

\begin{thm}[Strict local minimality]\label{thm-locminscale} Let $d\geq 1$ and $\mu,\nu \in \mathcal{P}_r(\R^d)$. Then we have:
\begin{enumerate}
\item If $L_0$ is a critical point of $L\mapsto \theta_L(\alpha)$ for all $\alpha>0$ and a strict local minimizer of $L\mapsto \theta_L(\alpha_0)$ in $\mathcal{D}_d$ for some $\alpha_0>0$, then $L_0$ is a strict local minimizer of $L\mapsto \theta_{\mu_{L}+\nu_{0}}(\alpha_0)$ in $\mathcal{D}_d$ at a low scale.
\item Let $L\in \mathcal{D}_d$.  If $z_0$ is a critical point of $z\mapsto \theta_{L+z}(\alpha)$ in $Q_L$ for any $\alpha>0$ and a strict local minimizer of $z\mapsto \theta_{L+z_0}(\alpha_0)$ for some $\alpha_0>0$, then $z_0$ is a strict local minimizer of $z\mapsto \theta_{\mu_L+\nu_z}(\alpha_0)$  on $Q_L$ at a low scale.
\end{enumerate}
The strict local minimality of $L_0$ and $z_0$ also holds at all scales if $\mu,\nu\in \mathcal{P}_r^{cm}(\R^d)$.
\end{thm}
\begin{proof}
Let us prove the first part of the theorem. The proof is actually a straightforward generalization of \cite[Prop. 11]{BetKnupfdiffuse}. According to \eqref{Jacobigeneral},  it is equivalent to show the strict local minimality of $L_0$ for
\begin{equation}\label{eq-hepsilondelta}
E_{h_{\varepsilon,\delta}}[L]:= \sum_{p\in L^*} h_{\varepsilon,\delta}(|p|^2),\quad h_{\varepsilon,\delta}(r)=\frac{g_{\mu^\varepsilon}(\sqrt{r})g_{\nu^\delta}(\sqrt{r})}{(\varepsilon \delta r)^{\frac{d}{2}-1}}e^{-\frac{\pi r}{\alpha_0}},
\end{equation}
where, for any measure $m\in \mathcal{P}_r(\R^d)$,
$$
g_{m^\varepsilon}(\sqrt{r})=\int_0^{\infty} J_{\frac{d}{2}-1}(4\pi s \varepsilon \sqrt{r})s^{1-\frac{d}{2}}d\psi_{m}(s).
$$
First, because $L_0$ is a critical point of $L\mapsto \theta_L(\alpha)$ in $\mathcal{D}_d$ for all $\alpha>0$, it implies, by Proposition \ref{prop-criticality}, that $L_0$ is a critical point of $E_{h_{\varepsilon,\delta}}$ in $\mathcal{D}_d$. Second, we also know that $D^2 E_{h_{0,0}}[L_0]=D^2 \theta_{L_0}(\alpha_0)$ is positive definite because $L_0$ is a strict local minimizer of $L\mapsto \theta_L(\alpha_0)$ in $\mathcal{D}_d$. Furthermore, all the coefficients of $D^2 E_{h_{\varepsilon,\delta}}[L_0]$ are expressed in terms of Bessel functions $J_m$ (see e.g. \cite[Eq. (9.1.30)]{AbraStegun}) and it is easy to check that $D^2 E_{h_{\varepsilon,\delta}}[L_0]=D^2 E_{h_{0,0}}[L_0]+A_{\varepsilon,\delta}[L_0]$, $A_{\varepsilon,\delta}[L_0]\in \mathcal{M}_d(\R)$, by the Taylor expansion of $J_m$ (see e.g. \cite[Eq. (9.1.10)]{AbraStegun}) where $\|A_{\varepsilon,\delta}[L_0]\|\to 0$ as $(\varepsilon,\delta)\to (0,0)$ for any chosen norm $\|.\|$ on $\mathcal{M}_d(\R)$. Indeed, this Taylor expansion gives 
$$
J_{\frac{d}{2}-1}(4\pi s\varepsilon |q|)J_{\frac{d}{2}-1}(4\pi t\delta |q|)=|q|^{d-2}\left(1+\varepsilon j_1(s,t,|q|)+\delta j_2(s,t,|q|)+j_{\varepsilon,\delta}(s,t,|q|)\right),
$$
where $j_1,j_2$ are independent of $\varepsilon,\delta$ and $j_{\varepsilon,\delta}$ is at least of order $\min\{\varepsilon,\delta\}^2$. Since $\psi_\mu$ and $\psi_\nu$ are both probability measures, the expansion of the second derivative is straightforward. Therefore, the result follows by continuity of $(\varepsilon,\delta)\mapsto A_{\varepsilon,\delta}[L_0]$ and by the fact -- following from the boundedness of the Bessel functions $J_m$ -- that $L\mapsto D^3 E_{h_{\varepsilon,\delta}}[L]$ is bounded on any ball centred at $L_0$, independently of $(\varepsilon,\delta)\in [0,1]^2$.

\medskip

Let us prove the second part of the theorem. Since $z_0$ is a critical point of $z\mapsto \theta_{L+z}(\alpha)$ in $Q_L$ for all $\alpha>0$, we have, by Proposition \ref{prop-criticality}, that $z_0$ is a critical point of $z\mapsto \theta_{\mu_L+\nu_z}(\alpha_0)$. By Proposition \ref{prop-Jacobigeneral}, we can write, as above,
\begin{equation*}
\theta_{\mu^\varepsilon_L+\nu^\delta_{z_0}}(\alpha_0)=\frac{\Gamma\hspace{-1mm}\left( \frac{d}{2} \right)^2}{4\alpha^{\frac{d}{2}}} \sum_{q\in L^*} h_{\varepsilon,\delta}(|q|^2)\cos(2\pi q\cdot z_0),\quad h_{\varepsilon,\delta}(r)=\frac{g_{\mu^\varepsilon}(\sqrt{r})g_{\nu^\delta}(\sqrt{r})}{(\varepsilon \delta r)^{\frac{d}{2}-1}}e^{-\frac{\pi r}{\alpha_0}}.
\end{equation*}
As explained in the proof of the first part of the theorem, by the Taylor expansion of $J_{d/2-1}$, it is straightforward to prove that 
\begin{align*}
D^2_z \theta_{\mu^\varepsilon_L+\nu^\delta_{z_0}}(\alpha_0)=D^2_z \theta_{L+z_0}(\alpha_0)+A_{\varepsilon,\delta,z_0}[L],
\end{align*}
where $A_{\varepsilon,\delta,z_0}[L]\in \mathcal{M}_d(\R)$ and $\sup_{z\in Q_L}\|A_{\varepsilon,\delta,z}[L]\|\to 0$ as $(\varepsilon,\delta)\to (0,0)$,  for any chosen norm $\|.\|$ on $\mathcal{M}_d(\R)$. Since $z_0$ is a strict local minimizer of $z\mapsto \theta_{L+z}(\alpha_0)$, it follows that $D^2_z \theta_{\mu^\varepsilon_L+\nu^\delta_{z_0}}(\alpha_0)$ is positive definite for $\varepsilon$ and $\delta$ sufficiently small, i.e. there exists $\varepsilon_0,\delta_0$ such that for any $0\leq \varepsilon<\varepsilon_0$ and any $0\leq \delta <\delta_0$, $z_0$ is a strict local minimizer of $z\mapsto  \theta_{\mu^\varepsilon_L+\nu^\delta_z}(\alpha_0)$.

\end{proof}
\begin{remark}\label{rmk-max}
We notice that if $L_0$ (resp. $z_0$) is a (strict local) minimizer of $L\mapsto \theta_{\mu_L+\nu_0}(\alpha_0)$ (resp. $z\mapsto \theta_{\mu_L+\nu_z}(\alpha_0)$) for any $\alpha_0$ belonging to a set of values $S$, therefore $\varepsilon_0$ and $\delta_0$ only depend on the maximum of these values.
\end{remark}

The next result is a generalization of \cite[Thm. 2 and 3]{BetKnupfdiffuse} in arbitrary dimension to our translated lattice theta function and gives a sufficient condition for the global minimality of a lattice in $\mathcal{D}_d$ or a point in $Q_L$.

\begin{thm}[Global minimality]\label{thm-mintransscale}
Let $d\geq 1$ and $\mu, \nu \in \mathcal{P}_r(\R^d)$, then we have:
\begin{enumerate}
\item  If $L_0$ is the unique global minimizer and a strict local minimizer of $L\mapsto \theta_L(\alpha)$ in $\mathcal{D}_d$ for all $\alpha>0$, then, for any $\alpha_0>0$, $L_0$ is the unique global minimizer of $L\mapsto \theta_{\mu_{L}+\nu_{0}}(\alpha_0)$ in $\mathcal{D}_d$ at a low scale.
\item Let $L\in \mathcal{D}_d$. If $z_0$ is a global minimizer and a strict local minimizer of $z\mapsto \theta_{L+z}(\alpha)$ in $Q_L$ for all $\alpha>0$, then, for any $\alpha_0>0$, $z_0$ is a global minimizer of $z\mapsto \theta_{\mu_L+\nu_z}(\alpha_0)$ in $Q_L$ at a low scale.
\end{enumerate}
The global minimality of $L_0$ and $z_0$ also hold at all scales if $\mu,\nu\in \mathcal{P}_r^{cm}(\R^d)$.
\end{thm}
\begin{proof}
Let us prove the first part of the theorem and let us start with the case $\mu,\nu\in \mathcal{P}_r^{cm}(\R^d)$. If $L_0$ is the unique global minimizer of $L\mapsto \theta_L(\alpha)$ for all $\alpha>0$, then, by Lemmas \ref{lem-L2} and \ref{lem-L4}, $L_0=L_0^*$ is the unique global minimizer of any lattice energy of the form $E_f$ where $f\in \mathcal{CM}_d$. Furthermore, it has been proved in \cite[Lem. 10]{BetKnupfdiffuse} that $\mu \in \mathcal{P}_r^{cm}(\R^d)\iff\hat{\mu} \in \mathcal{P}_r^{cm}(\R^d)$. Therefore, as the set of completely monotone functions is stable by product, we have $\hat{\mu}\hat{\nu}\in \mathcal{P}_r^{cm}(\R^d)$ and it follows from the complete monotonicity of $t\mapsto e^{-\beta t}$ for any $\beta>0$ that $h:t\mapsto e^{-\pi t/\alpha_0}t^{1-d/2} g_\mu(\sqrt{t})g_\nu(\sqrt{t})$ is completely monotone. Therefore, $L_0$ is the unique global minimizer of $L\mapsto \theta_{\mu_L+\nu_0}(\alpha_0)$ for any fixed $\alpha_0>0$.

\medskip

For $\mu,\nu\in \mathcal{P}_r(\R^d)$, the proof  is again a generalization of \cite[Thm. 2]{BetKnupfdiffuse}. We first remark that any minimizer of $L\mapsto \theta_{\mu^\varepsilon_L+\nu^\delta_0}(\alpha_0)$ belongs to a ball of center $L_0$ with a finite radius. Indeed, this fact is proved in \cite[Lemma 12]{BetKnupfdiffuse} in dimension $d=2$ and is directly generalizable to any dimension by bounding below $|p|^2$, $p\in L$, by $|p'|^2$ where $p'\in L_\lambda=\lambda^{-1}\Z\oplus \Lambda$ for some $\lambda \geq 1$ which is also a parameter of $L$, once the lattice is parametrized by $n_d$ parameters as explained in Section \ref{section-def}, and $\Lambda\in \mathcal{L}_{d-1}^\circ$. We therefore get

$$
\theta_{\mu^\varepsilon_L+\nu^\delta_0}(\alpha_0)\geq \theta_{\mu^\varepsilon_{L_\lambda}+\nu^\delta_0}(\alpha_0)\geq \iint_{\R^d\times \R^d} \sum_{m\in \Z} e^{-\pi\alpha_0|\lambda^{-1} (m,0,...,0) +x-y|^2}d\mu^\varepsilon(x)d\nu^{\delta}(y)\to +\infty
$$
as $\lambda \to +\infty$. The same can be done in all the unbounded directions of the fundamental domain $\mathcal{D}_d$, proving that the global minimizer necessarily belongs to a compact subset of $\mathcal{D}_d$, for example a closed ball containing $L_0$ that we will note $K$. 

\medskip

In the following, we write, as in \eqref{eq-hepsilondelta},
\begin{equation*}
E_{h_{\varepsilon,\delta}}[L]:= \sum_{p\in L} h_{\varepsilon,\delta}(|p|^2),\quad h_{\varepsilon,\delta}(r)=\frac{g_{\mu^\varepsilon}(\sqrt{r})g_{\nu^\delta}(\sqrt{r})}{(\varepsilon \delta r)^{\frac{d}{2}-1}}e^{-\frac{\pi r}{\alpha_0}},
\end{equation*}
and we recall that minimizing $L\mapsto \theta_{\mu^\varepsilon_L+\nu^\delta_0}(\alpha_0)$ is equivalent with minimizing $L\mapsto E_{h_{\varepsilon,\delta}}$ because, by Lemma \ref{lem-L2}, $L_0=L_0^*$. As in \cite{BetKnupfdiffuse}, we claim there exists $\widetilde{h}_{\varepsilon,\delta}$ such that, for any $L\in K$,
\begin{align}
&\left|  E_{\widetilde{h}_{\varepsilon,\delta}}[L] -E_{h_{\varepsilon,\delta}}[L] \right|\leq C \max\{\varepsilon,\delta\}^2,\quad \textnormal{as $\varepsilon,\delta\to 0$}, \label{eq-approxL01} \\
& \quad E_{\widetilde{h}_{\varepsilon,\delta}}[L]- E_{\widetilde{h}_{\varepsilon,\delta}}[L_0]\geq C d(L,L_0)^2,\label{eq-approxL02}
\end{align}
for some constant $C>0$ independent of $\varepsilon,\delta$. If \eqref{eq-approxL01}-\eqref{eq-approxL02} hold, then, for any $L\in K$,
$$
E_{h_{\varepsilon,\delta}}[L]-E_{h_{\varepsilon,\delta}}[L_0]\geq C d(L,L_0)^2-2C\max\{\varepsilon,\delta\}^2.
$$
Thus, for any $L\in K$ such that $E_{h_{\varepsilon,\delta}}[L]\leq E_{h_{\varepsilon,\delta}}[L_0]$, this implies that $d(L,L_0)\leq C\max\{\varepsilon,\delta\}$ which contradicts the strict local minimality of $L_0$ for sufficiently small $\varepsilon$ and $\delta$. To construct $\tilde{h}_{\varepsilon,\delta}$, the idea is to approximate 
$$
k_{\varepsilon,\delta}(r):=h_{\varepsilon,\delta}(r)e^{\frac{\pi r}{\alpha_0}}=k^1_\varepsilon(r)k^2_\delta(r),\quad k^1_\varepsilon(r):=\frac{g_{\mu^\varepsilon}(\sqrt{r})}{(\varepsilon\sqrt{r})^{\frac{d}{2}-1}},\quad k^2_\delta(r):=\frac{g_{\nu^\delta}(\sqrt{r})}{(\delta\sqrt{r})^{\frac{d}{2}-1}},
$$
by a bounded completely monotone function $\widetilde{k}_{\varepsilon,\delta}=\mathcal{L}[m_{\varepsilon,\delta}]$ such that $m_{\varepsilon,\delta}$ is a positive measure with a compact support, $\widetilde{k}_{\varepsilon,\delta}(0)=1$ and $\|\widetilde{k}_{\varepsilon,\delta}'\|\leq C$. It is indeed sufficient to apply the method described in \cite[proof of Thm. 2]{BetKnupfdiffuse} to $k^1_\varepsilon$ and $k^2_\delta$ that are then approximated respectively by $\widetilde{k}_\varepsilon^1=\mathcal{L}[m_\varepsilon]$ and $\widetilde{k}_\varepsilon^2=\mathcal{L}[m_\delta]$ and where $m_{\varepsilon,\delta}=m_\varepsilon \ast m_\delta$. We therefore get, for any $r>0$,
\begin{equation*}
\left| k_{\varepsilon,\delta}(r)- \widetilde{k}_{\varepsilon,\delta}(r)\right|\leq C\min\{ \max\{\varepsilon,\delta\}^2r,1 \}.
\end{equation*}
Defining $\tilde{h}_{\varepsilon,\delta}(r):=\tilde{k}_{\varepsilon,\delta}(r)e^{-\frac{\pi r}{\alpha_0}}$, we then have
\begin{align*}
\left|  E_{\widetilde{h}_{\varepsilon,\delta}}[L] -E_{h_{\varepsilon,\delta}}[L] \right|&\leq\sum_{q\in L^*} e^{-\frac{\pi |q|^2}{\alpha_0}}\left| k_{\varepsilon,\delta}(|q|^2)- \widetilde{k}_{\varepsilon,\delta}(|q|^2)\right|\\
&\leq C\sum_{q\in L^*} \min\{ \max\{\varepsilon,\delta\}^2|q|^2,1 \}e^{-\frac{\pi |q|^2}{\alpha_0}}\\
&\leq C\max\{\varepsilon,\delta\}^2,
\end{align*}
by the exponential decay of the lattice theta function and where $C$ does not depend on the lattice $L\in K$. Therefore, \eqref{eq-approxL01} is proved. For the second inequality, we have
\begin{align*}
E_{\widetilde{h}_{\varepsilon,\delta}}[L]- E_{\widetilde{h}_{\varepsilon,\delta}}[L_0]\geq Cd(L,L_0)^2,
\end{align*}
for a constant $C$ independent of $\varepsilon$ and $\delta$. Indeed, it is a consequence of the fact that $\tilde{h}_{\varepsilon,\delta}=\mathcal{L}[m_{\varepsilon,\delta}]$ and $m_{\varepsilon,\delta}$ has a compact support and then follows by the strict local minimality of $L_0$ for $L\mapsto \theta_L(\alpha)$ for all $\alpha>0$, which implies the same strict local minimality for any $E_f$ where $f\in \mathcal{CM}_d$ (see Lemma \ref{lem-L4}), and in particular for $E_{\widetilde{h}_{\varepsilon,\delta}}$ (see \cite[Section 2.2]{BetKnupfdiffuse} for details).

\medskip

Let us prove the second part of the theorem and let us begin again with $\mu,\nu\in \mathcal{P}_r^{cm}(\R^d)$. The proof uses the same ingredient as the previous one. Indeed, if $z_0$ is a global minimizer of $z\mapsto \theta_{L+z}(\alpha)$ in $Q_L$ for all $\alpha>0$, then by Lemma \ref{lem-L4} the same holds for
$$
z\mapsto E_f[L+z]=\sum_{p\in L} f(|p+z|^2),
$$
where $f\in \mathcal{CM}_d$. Let $\alpha_0>0$, then by the generalized Jacobi Transformation Formula \eqref{Jacobigeneral}, we have, for some constant $C_{\alpha_0,d}>0$,
$$
\theta_{\mu_L+\nu_z}(\alpha_0) =C_{\alpha_0,d}\sum_{q\in L^*} h(|q|^2)e^{2i\pi q\cdot z},
$$
for a completely monotone function $h\in \mathcal{CM}_d$, as explained previously. We therefore obtain, by Poisson Summation Formula (see e.g. \cite[Appendix A]{SaffLongRange}),
$$
\theta_{\mu_L+\nu_z}(\alpha_0) =\tilde{C}_{\alpha_0,d}\sum_{p\in L} \hat{h}(|p+z|^2).
$$
Since $h$ is completely monotone, then $\hat{h}$ is completely monotone by \cite[Lem. 10]{BetKnupfdiffuse} and therefore $z_0$ is a global minimizer of $z\mapsto \sum_{p\in L} \hat{h}(|p+z|^2)$ by Lemma \ref{lem-L4}, which concludes the proof.

\medskip

Let us now consider $\mu,\nu\in \mathcal{P}_r(\R^d)$. For convenience, we define, for fixed $\alpha_0>0$ and $L\in \mathcal{D}_d$,
$$
F_{h_{\varepsilon,\delta}}(z):=\sum_{q\in L^*} h_{\varepsilon,\delta}(|q|^2)\cos(2\pi q\cdot z),
$$
where $h_{\varepsilon,\delta}$ is defined by \eqref{eq-hepsilondelta}. We again recall that minimizing $z\mapsto \theta_{\mu_L+\nu_z}(\alpha_0)$ is equivalent with minimizing $F_{h_{\varepsilon,\delta}}$. We remark that, for $\widetilde{h}_{\varepsilon,\delta}$ defined above, we have, for any $z\in Q_L$,
\begin{align*}
&\left|  F_{\widetilde{h}_{\varepsilon,\delta}}(z) -F_{h_{\varepsilon,\delta}}(z) \right|\leq C \max\{\varepsilon,\delta\}^2,\quad \textnormal{as $\varepsilon,\delta\to 0$} \\
& \quad F_{\widetilde{h}_{\varepsilon,\delta}}(z)- F_{\widetilde{h}_{\varepsilon,\delta}}(z_0)\geq C |z-z_0|^2,
\end{align*}
for some constant $C>0$. The second inequality follows from the strict local minimality of $z_0$ for $z\mapsto \theta_{L+z}(\alpha_0)$. Therefore, for any $z\in Q_L$,
$$
F_{h_{\varepsilon,\delta}}(z)-F_{h_{\varepsilon,\delta}}(z_0)\geq C|z-z_0|^2 -2C\max\{\varepsilon,\delta\}^2.
$$
Thus, for any $z\in Q_L$ such that $F_{h_{\varepsilon,\delta}}(z)\leq F_{h_{\varepsilon,\delta}}(z_0)$, this implies $|z-z_0|\leq C\max\{\varepsilon,\delta\}$ which is not true for sufficiently small $\varepsilon$ and $\delta$ by the strict local optimality of $z_0$ in $Q_L$ previously shown in Theorem \ref{thm-locminscale}.
\end{proof}
\begin{remark}[Difference between $\varepsilon$ and $\delta$]\label{rmksize} An interesting property is the fact that $\varepsilon$ and $\delta$ can be chosen with different scales -- for instance $\varepsilon=1/n$ and $\delta=1/\sqrt{n}$ --  as long as they are below the critical values $\varepsilon_0$ and $\delta_0$. 
\end{remark}

\begin{remark}[Global minimizer of the lattice theta function]
Notice that the minimizer of $L\mapsto \theta_L(\alpha)$ in $\mathcal{D}_d$ is known to be the same for all $\alpha>0$, so far, only in dimension $d\in \{2,8,24\}$ as proved in \cite{Mont,CKMRV2Theta}. In these dimensions, the minimizers are, respectively, the triangular lattice $\Lambda_1$, $\mathsf{E}_8$ and the Leech lattice.
\end{remark}

\begin{remark}[Global minimizer of the translated lattice theta function]
Few results are already known concerning the minimization of $z\mapsto \theta_{L+z}(\alpha)$, where $L\in \mathcal{D}_d$ is fixed:
\begin{enumerate}
\item In dimension $d=2$, if $L=\Lambda_1$, then Baernstein \cite{Baernstein-1997} proved that the barycenters $z_1=\sqrt{\frac{2}{\sqrt{3}}}\left(\frac{1}{2},\frac{1}{2\sqrt{3}}  \right)$ and $z_2=\sqrt{\frac{2}{\sqrt{3}}}\left( 1,\frac{1}{\sqrt{3}} \right)$ of the primitive triangle composing $Q_{\Lambda_1}$ are the unique global minimizers of $z\mapsto \theta_{\Lambda_1+z}(\alpha)$ for all $\alpha>0$.
\item For any $d\geq 1$, we proved in \cite[Prop. 1.3]{BeterminPetrache} (see also \cite[Prop. 3.7]{BeterminKnuepfer-preprint}) that the center $c_a=\frac{1}{2}(a_1,...,a_d)$ of the primitive cell $Q_{\Z^d_a}$ of the orthorhombic lattice $\Z^d_a$ defined by \eqref{ortho} is the unique global minimizer of $z\mapsto \theta_{\Z^d_a+z}(\alpha)$ for all $\alpha>0$.
\end{enumerate}
Applications of these results will be shown in the next section, as well as the optimality of the center of the primitive hexagon in the honeycomb lattice case, using Baernstein's theorem.

\medskip

\noindent Furthermore, if $z_0$ is a global minimizer of $z\mapsto \theta_{L+z}(\alpha)$ for all $\alpha>0$, then it is the case as $\alpha\to +\infty$ and it turns out that $z_0$ is necessarily a deep hole of $L$, i.e. a solution of the following optimization problem:
\begin{equation}\label{bighole}
\max_{z\in \R^d}\min_{p\in L} |z-p|,
\end{equation}
as we proved in \cite[Thm 1.5]{BeterminPetrache}. Therefore, since this property does not necessarily hold as $\alpha\to 0$ (see e.g. in dimension $d=2$ in \cite[Thm 1.6]{BeterminPetrache} for asymmetric lattices), the global minimizer of $z\mapsto \theta_{L+z}(\alpha)$ is not necessarily the same for all $\alpha$.
\end{remark}

\section{Applications to particular lattices}\label{sect-appli}

We finally apply the previous results to some particular lattices defined by \eqref{Lambda1}-\eqref{FCC} as well as $\mathsf{D}_4$, $\mathsf{E}_8$, the Leech lattice and $\mathsf{D}_d^+$. This is the perfect opportunity to recall the main results that are currently known about the local and global minima of the (translated) lattice theta functions, which are now generalized to the masses interactions case.

\medskip

Montgomery \cite{Mont} proved the minimality of $\Lambda_1$ in $\mathcal{D}_2$ for $L\mapsto \theta_L(\alpha)$ for all $\alpha>0$. Furthermore, as recalled in the previous section, Baernstein \cite{Baernstein-1997} proved the minimality of the two barycenters of the primitive triangles composing $Q_{\Lambda_1}$ for $z\mapsto \theta_{\Lambda_1+z}(\alpha)$ also for all $\alpha>0$.  Therefore, applying Montgomery's theorem, Baernstein's theorem, and our main results, we extend the minimality of $\Lambda_1$ and its barycenters to Gaussian masses interactions. Notice that the $\mu=\nu, z=0$ case has been already proved in \cite{BetKnupfdiffuse}.

\begin{corollary}[The triangular lattice]\label{cor-triangular}
Let $d=2$ and $\Lambda_1$ be defined by \eqref{Lambda1}. We then have:
\begin{enumerate}
\item  For any $\mu,\nu \in \mathcal{P}_r(\R^2)$, $\Lambda_1$ is a critical point of $L\mapsto \theta_{\mu_L+\nu_0}(\alpha)$ for all $\alpha>0$. Moreover, for all $\alpha_0>0$, $\Lambda_1$ is a global minimizer of $L\mapsto \theta_{\mu_L+\nu_0}(\alpha_0)$ in $\mathcal{D}_2$ at a low scale. Furthermore, if $\mu,\nu \in \mathcal{P}_r^{cm}(\R^2)$, then the global minimality holds at all scales.

\item Let $z_1=\sqrt{\frac{2}{\sqrt{3}}}\left(\frac{1}{2},\frac{1}{2\sqrt{3}}  \right)$ and $z_2=\sqrt{\frac{2}{\sqrt{3}}}\left( 1,\frac{1}{\sqrt{3}} \right)$ be the barycenters of the primitive triangles of $Q_{\Lambda_1}$, then, for any $\mu,\nu \in \mathcal{P}_r(\R^2)$, $z_1$ and $z_2$ are critical points of $z\mapsto \theta_{\mu_{\Lambda_1}+\nu_z}(\alpha)$ in $Q_{\Lambda_1}$ at all scales. Furthermore, for any $\alpha_0>0$, $z_1$ and $z_2$ are the unique global minimizers of $z\mapsto \theta_{\mu_{\Lambda_1}+\nu_z}(\alpha_0)$ in $Q_{\Lambda_1}$ at a low scale. Moreover, if  $\mu,\nu \in \mathcal{P}_r^{cm}(\R^2)$, then the minimality of $z_1$ and $z_2$ holds at all scales.
\end{enumerate}
\end{corollary}

\begin{remark}[Importance of the triangular lattice]
The triangular lattice arises in many physical models such as Bose-Einstein Condensates \cite{AftBN}, Superconductivity \cite{Sandier_Serfaty}, Coulomb Gases \cite{2DSandier} or diblock copolymer interaction \cite{CheOshita}. We also recall that $\Lambda_1$ is conjectured (see Cohn and Kumar \cite[Conjecture 9.4]{CohnKumar}) to be the unique minimizer of the lattice theta function among periodic configurations (not only among Bravais lattices) of fixed density. Furthermore, $\Lambda_1$ can also be viewed as a layer of a FCC or BCC lattice potentially shifted by a vector parallel to $z_1$ or $z_2$, as explained in \cite[Sect. II.2]{BeterminPetrache}. Then, as already discussed in \cite{Beterminlocal3d}, Corollary \ref{cor-triangular} is of great interest for the understanding of BCC and FCC stability in the smeared out particle case, using dimension reduction techniques as in \cite{BeterminPetrache}.
\end{remark}
\medskip

We now show the following result, using Baernstein's theorem \cite[Thm. 1]{Baernstein-1997}, for the honeycomb lattice $\mathsf{H}$.

\begin{prop}[The honeycomb lattice]\label{cor-honey}
Let $d=2$, and $\mathsf{H}$ be defined by \eqref{def-honeycomb} with the primitive hexagon $\mathcal{H}$ containing the center $z_0=\sqrt{\frac{2}{\sqrt{3}}}\left(\frac{1}{2},\frac{1}{2\sqrt{3}}\right)$. Then, for any $\mu,\nu \in \mathcal{P}_r(\R^2)$, $z_0$ is a critical point of $z\mapsto \theta_{\mu_{\mathsf{H}}+\nu_z}(\alpha)$ in $\mathcal{H}$ at all scales. Furthermore, $z_0$ is the unique global minimizer of $z\mapsto \theta_{\mu_{\mathsf{H}}+\nu_z}(\alpha)$ in $\mathcal{H}$ at a low scale. Moreover, if $\mu,\nu \in \mathcal{P}_r^{cm}(\R^2)$, then the minimality holds at all scales.
\end{prop}
\begin{proof}
It is sufficient to show the $\mu=\nu=\delta_0$ case and to apply the second part of Theorem \ref{thm-mintransscale} which remains true in this non-Bravais case by simply following the lines of its proof and using the fact that
\begin{equation*}
\theta_{\mathsf{H}+z}(\alpha)=\frac{1}{2}\left(\theta_{\Lambda_1+z}(\alpha)+\theta_{\Lambda_1+u+z}(\alpha)\right).
\end{equation*}
Since $z\mapsto \theta_{\Lambda_1+z}(\alpha)$ have two global minimizers in $Q_{\Lambda_1}$ that are the barycenters of the primitive triangle of $Q_{\Lambda_1}$, we obtain that, for any $z\in \mathcal{H}$,
$$
\theta_{\Lambda_1+z}(\alpha)\geq \theta_{\Lambda_1+z_1}(\alpha).
$$
 Furthermore, by symmetry, $z_1+u=z_2$ is the second gobal minimizer of $z\mapsto \theta_{\Lambda_1+z}(\alpha)$ and we get, for any $z\in \mathcal{H}$,
$$
\theta_{\Lambda_1+u+z}(\alpha)\geq \theta_{\Lambda_1+z_2}(\alpha)=\theta_{\Lambda_1+u+z_1}(\alpha).
$$
We therefore have proved that
$$
\theta_{\mathsf{H}+z}(\alpha)\geq \theta_{\mathsf{H}+z_1}(\alpha),
$$
and $z_1$ is the unique global minimizer of $z\mapsto \theta_{\mathsf{H}+z}(\alpha)$ in $\mathcal{H}$ for all $\alpha>0$, by symmetry. 
\end{proof}

\begin{figure}[!h]
\centering
\includegraphics[width=12cm]{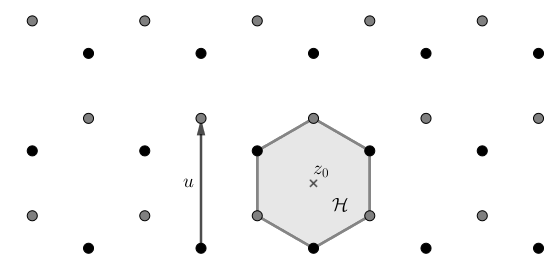} 
\caption{Minimizer $z_0$ of $z\mapsto \theta_{\mathsf{H}+z}(\alpha)$ in a primitive hexagon $\mathcal{H}$. The black dots represent $\Lambda_1$ and the grey one represent $\Lambda_1+u$.}
\label{Honeycomb}
\end{figure}

\begin{remark}[Gaussian interaction between Lithium atom and Graphene sheet]
The Gaussian interaction between a mass and a honeycomb structure is physically relevant. For instance, in \cite{GaussianGraphene}, the authors have been designed a Gaussian approximation potential modelling the interaction energy between a Lithium atom and a Graphene sheet structure. This has been done by applying Machine Learning to Density Functional Theory reference data. Our result gives the exact location of the Lithium atom minimizing the Gaussian interaction with a Graphene sheet.
\end{remark}

Because of the particular role of the cubic lattices $\Z^d$ and their orthorhombic deformations $\Z^d_a$ defined by \eqref{ortho}, we have summarized all the results related to them in the following corollary. The first point is an easy consequence of \cite[Prop. 5.1]{Beterminlocal3d}, the second point is based on \cite[Thm 2]{Mont} and the third point follows from \cite[Prop. 1.3]{BeterminPetrache}.

\begin{corollary}[Cubic and Orthorhombic lattices]
Let $d\geq 2$, then
\begin{enumerate}
\item For any $\mu,\nu\in \mathcal{P}_r(\R^d)$, $\Z^d$ is a critical point of $L\mapsto \theta_{\mu_L+\nu_0}(\alpha)$ in $\mathcal{D}_d$ at all scales. 
\item For any $\mu,\nu\in \mathcal{P}_r(\R^d)$ and any $\alpha_0>0$, $a=(1,1,...,1)$ is the unique minimizer of $a\mapsto \theta_{\mu_{\Z^d_a}+\nu_0}(\alpha_0)$ in $\{(a_1,...,a_d)\in (0,+\infty) : \prod_{i=1}^d a_i=1  \}$ at a low scale, i.e. $\Z^d$ is the unique minimizer of $L\mapsto \theta_{\mu_L+\nu_0}(\alpha_0)$ among orthorhombic lattices at a low scale. Moreover, if $\mu,\nu\in \mathcal{P}_r^{cm}(\R^d)$, the minimality holds at all scales.
\item Let $(a_1,...,a_d)\in (0,+\infty)$ be such that $\prod_{i=1}^d a_i=1$ and $c_a:=\frac{1}{2}(a_1,...,a_d)$ be the center of $Q_{\Z^d_a}$. Then, for any $\mu,\nu\in \mathcal{P}_r(\R^d)$,  $c_a$  is a critical point of $z\mapsto \theta_{\mu_{\Z^d_a}+\nu_z}(\alpha)$ in $Q_{\Z^d_a}$ at all scales. Furthermore, for any $\alpha_0>0$, $c_a$ is the unique global minimizer of $z\mapsto \theta_{\mu_{\Z^d_a}+\nu_z}(\alpha_0)$ in $Q_{\Z^d_a}$  at a low scale. Moreover, if  $\mu,\nu\in \mathcal{P}_r^{cm}(\R^d)$, the minimality holds at all scales.
\end{enumerate}
\end{corollary}

\begin{remark}[Octahedral site of the cubic lattice]\label{rmk:octa}
In the cubic case $(a_1,...,a_d)=(1,...,1)$, the center $c=(1/2,...,1/2)$, also called the octahedral site of $\Z^d$, is indeed the preferred location to add an extra atom in order to create, for example, an ion (like the CsCl which has a BCC structure).
\end{remark}

In \cite{BeterminPetrache}, it has been proved, using a computer assisted method, that $\mathsf{D}_3^*$ (resp. $\mathsf{D}_3$) is a strict local minimizer of $L\mapsto \theta_L(\alpha)$ for any $\alpha\in \mathcal{A}$ (resp. $\alpha \in \mathcal{B}$) defined by
\begin{equation}\label{eq-A}
\mathcal{A}:=\{0.001k: k\in \N, 1\leq k\leq 1000  \},\quad \mathcal{B}:=\left\{ 1/x: x\in\mathcal A \right\}.
\end{equation}
We also proved in \cite[Thm. 1.3]{Beterminlocal3d} that the same result holds for extremal values of $\alpha$, i.e. if $\alpha>\alpha_1$ (resp. $0<\alpha<\alpha_1^{-1}$) for $\mathsf{D}_3$ (resp. $\mathsf{D}_3^*$), for some $\alpha_1>1$. Moreover,we also add that $\mathsf{D}_3$ and $\mathsf{D}_3^*$ appear to be saddle points for the lattice theta function respectively for $\alpha<\alpha_2^{-1}$ and $\alpha>\alpha_2$. Applying Theorem \ref{thm-locminscale} to these specific $\alpha$, we get the strict local minimality of these lattices for measures that are sufficiently concentrated around the lattice points, where the threshold values $\varepsilon_0$ and $\delta_0$ only depend on $\mu,\nu$, the maximum of $\mathcal{A}$ and $\mathcal{B}$ or on $\alpha_1$ (see Remark \ref{rmk-max}).

\medskip

\begin{corollary}[The FCC and BCC lattices]\label{cor:BCCFCC}
Let $\mu,\nu\in \mathcal{P}_r(\R^3)$, we have:
\begin{enumerate}
\item The lattices $\mathsf{D}_3$ and $\mathsf{D}_3^*$ are critical points of $L\mapsto \theta_{\mu_L+\nu_0}(\alpha)$ in $\mathcal{D}_3$ at all scales.
\item Let $\mathcal{A},\mathcal{B}$ be defined by \eqref{eq-A}. For any $\alpha_0\in \mathcal{A}$ (resp. $\alpha_0\in \mathcal{B}$), $\mathsf{D}_3^*$ (resp. $\mathsf{D}_3$) is a strict local minimizer of $L\mapsto \theta_{\mu_L+\nu_0}(\alpha_0)$ on $\mathcal{D}_3$ at a low scale. Furthermore, $\varepsilon_0$ and $\delta_0$ only depend on $\mu,\nu$, $a=\max \mathcal{A}$ and $b=\max \mathcal{B}$. 
\item There exists $\alpha_1>1$ such that for any $\alpha_0\in (\alpha_1,+\infty)$ (resp. $\alpha\in (0,\alpha_1^{-1})$), $\mathsf{D}_3$ (resp. $\mathsf{D}_3^*$) is a strict local minimizer of $L\mapsto \theta_{\mu_{L}+\nu_{0}}(\alpha_0)$ at a low scale. In the $\mathsf{D}_3^*$ case, $\varepsilon_0$ and $\delta_0$ only depend on $\mu,\nu$ and $\alpha_1$.
\end{enumerate}
\end{corollary}

\begin{remark}[Conjecture for $\mathsf{D}_3$ and $\mathsf{D}_3^*$ in the soft case]
We also notice that, according to Sarnak and Str\"ombergsson \cite[Eq. (43)]{SarStromb}, $\mathsf{D}_3$ (resp. $\mathsf{D}_3^*$) is expected to be a global minimizer of $L\mapsto \theta_L(\alpha)$ in $\mathcal{D}_3$ for any $\alpha\geq 1$ (resp. $\alpha\leq 1$). Corollary \ref{cor:BCCFCC} and the fact that the only three-dimensional ``volume-stationary" lattices are $\Z^3,\mathsf{D}_3$ and $\mathsf{D_3}^*$ (see \cite[Section 3]{LBMorse}) supports this conjecture for the soft lattice theta function.
\end{remark}

\begin{remark}[Minimality of $\mathsf{D}_3$ and $\mathsf{D}_3^*$ among body-centred-orthorhombic lattices]
Based on \cite[Thm 1.7]{BeterminPetrache}, it is also straightforward to prove the global minimality of $\mathsf{D}_3^*$ (respectively $\mathsf{D}_3$) among body-centred-orthorhombic lattices (resp. their dual lattices), which are the anisotropic dilations of the BCC lattice (based on the unit cube) along the coordinate axes by $\sqrt{y},1/\sqrt{y}$ and $t$. They are defined by
$$
L_{y,t}:=\bigcup_{k\in \Z} \left( \Z(\sqrt{y},0) \oplus \Z\left(0,\frac{1}{\sqrt{y}}\right) + (\sqrt{y}/2,1/2 \sqrt{y},0)\mathds{1}_{2\Z}(k) +k(0,0,t/2) \right),
$$ 
and $\mathsf{D}_3^*$ (resp. $\mathsf{D}_3$) is the unique minimizer of $L\mapsto \theta_{\mu_L+\nu_0}(\alpha_0)$ in this class of lattices where $\alpha_0\in \mathcal{A}$ (resp. $\mathcal{B}$) and $t=1$, at a low scale.
\end{remark}

\begin{remark}[Deep holes and BCC/FCC lattices]
For $L\in \{\mathsf{D}_3,\mathsf{D}_3^*\}$, the global minimizers of $z\mapsto \theta_{L+z}(\alpha)$ in $Q_L$ are expected to be the deep holes of $L$, solution of \eqref{bighole}, i.e. $z_0=2^{-\frac{1}{3}}(1,1,1)$ for $\mathsf{D}_3$ and $z_1=2^{-\frac{2}{3}}(1,1,0)$ for $\mathsf{D}_3^*$ as well as all their images in $Q_L$ by symmetry. These locations are the usual one where a different atom can be added to the structure, in order to create, for example, an ion. They are also called octahedral sites, as for the cubic lattice (see Remark \ref{rmk:octa}).
\end{remark}

Finally, the same kind of local results for $L\mapsto \theta_{\mu_L+\nu_0}(\alpha)$ can be stated for some other special lattices, based on  \cite{SarStromb,Coulangeon:2010uq,CoulSchurm2018}.

\begin{corollary}[Dimensions $d\geq 4$ -- Special cases] We have that:
\begin{enumerate}
\item $\mathsf{D}_4$, $\mathsf{E}_8$ and the Leech lattice are strict local minimizers of $L\mapsto \theta_{\mu_L+\nu_0}(\alpha_0)$ on $\mathcal{D}_d$, for all $\alpha_0>0$, at a low scale if $\mu,\nu\in \mathcal{P}_r(\R^d)$ and at all scales if $\mu,\nu \in \mathcal{P}_r^{cm}(\R^d)$, for the corresponding dimensions $d\in \{4,8,24\}$.
\item For all odd integer $d\geq 9$, there exists $\alpha_d$ such that for any $\alpha_0>\alpha_d$, $\mathsf{D}_d^+$ is a strict local minimizer of $L\mapsto \theta_{\mu_L+\nu_0}(\alpha_0)$ on $\mathcal{D}_d$ at a low scale if  $\mu,\nu\in \mathcal{P}_r(\R^d)$.
\item $\mathsf{E}_8$ and the Leech lattice are global minimizers of $L\mapsto \theta_{\mu_L+\nu_0}(\alpha_0)$ on $\mathcal{D}_d$, for any $\alpha_0>0$, at a low scale if $\mu,\nu\in \mathcal{P}_r(\R^d)$ and at all scales if $\mu,\nu \in \mathcal{P}_r^{cm}(\R^d)$, for the corresponding dimensions $d\in \{8,24\}$
\end{enumerate} 
\end{corollary}


\medskip

\noindent \textbf{Acknowledgement.} I acknowledge support
from VILLUM FONDEN via the QMATH Centre of Excellence (grant No. 10059). Furthermore, I would like to thank the anonymous referee for her/his useful comments.

{\small
\bibliographystyle{plain} 
\bibliography{Masstheta}}

\end{document}